\documentclass[a4paper,english]{lipics-v2021}

\usepackage{xspace}
\usepackage{csquotes}
\usepackage{complexity}

\renewcommand{\H}{\ensuremath{\mathbb{H}}\xspace}
\renewcommand{\R}{\ensuremath{\mathbb{R}}\xspace}
\renewcommand{\D}{\ensuremath{\mathcal{D}}\xspace}
\renewcommand{\A}{\ensuremath{\mathcal{A}}\xspace}

\newcommand{\ER}{\ensuremath{\exists\R}\xspace}
\newcommand{\problemname}[1]{{\normalfont\textsc{#1}}\xspace}
\newcommand{\graphclass}[1]{{\normalfont\textsf{#1}}\xspace}
\newcommand{\ETR}{\problemname{ETR}}
\newcommand{\Stretchability}{\problemname{Stretchability}}
\newcommand{\SimpleStretchability}{\problemname{SimpleStretchability}}

\newcommand{\Poincare}{Poincar{\'{e}}\xspace}
\DeclareMathOperator{\arccosh}{arcosh}
\newcommand{\abs}[1]{\lvert #1 \rvert}

\title{Recognizing Unit Disk Graphs in Hyperbolic Geometry is \texorpdfstring{$\boldsymbol{\exists\mathbb{R}}$}{ER}-Complete}

\author{Nicholas Bieker}{Karlsruhe Institute of Technology, Germany}{bieker.nicholas@gmail.com}{}{}

\author{Thomas Bl\"{a}sius}{Karlsruhe Institute of Technology, Germany}{thomas.blaesius@kit.edu}{}{}

\author{Emil Dohse}{Karlsruhe Institute of Technology, Germany}{emildohse@gmail.com}{}{}

\author{Paul Jungeblut}{Karlsruhe Institute of Technology, Germany}{paul.jungeblut@kit.edu}{0000-0001-8241-2102}{}

\authorrunning{N. Bieker, T. Bl{\"{a}}sius, E. Dohse and P. Jungeblut}

\Copyright{Nicholas Bieker, Thomas Bl\"{a}sius, Emil Dohse, Paul Jungeblut}



\hideLIPIcs
\nolinenumbers

\begin{document}

\maketitle

\begin{abstract}
    A graph~$G$ is a (Euclidean) unit disk graph if it is the intersection graph of unit disks in the Euclidean plane~$\R^2$.
    Recognizing them is known to be \ER-complete, i.e., as hard as solving a system of polynomial inequalities.
    In this note we describe a simple framework to translate \ER-hardness reductions from the Euclidean plane~$\R^2$ to the hyperbolic plane~$\H^2$.
    We apply our framework to prove that the recognition of unit disk graphs in the hyperbolic plane is also \ER-complete.
\end{abstract}

\section{Introduction}

A graph is a \emph{unit disk graph} if its vertices can be represented by equally sized disk such that two vertices are adjacent if and only if their corresponding disks intersect.
The class of unit disk graphs (\graphclass{UDG}) is a well studied graph class due to its mathematical beauty and its practical relevance, e.g., in the context of sensor networks.

Naturally, unit disk graphs are usually considered in the Euclidean plane~$\R^2$.
However, in the past decade, research on intersection graphs of equally sized disks in the hyperbolic plane~$\H^2$ has gained traction.
This is due to the fact that the hyperbolic geometry is well suited to represent a wider range of graph structures, including complex scale-free networks with heterogeneous degree distributions~\cite{Blasius2018_HRGCliques,Bode2015_HRGComponent,Gugelmann2012_HRGDegSequenceClustering,Krioukov2010_HRG,Muller2019_HRGDiameter}; see Figure~\ref{fig:unit_disk_graphs}.
Most research on such graphs is driven by the network science community studying probabilistic network models, i.e., hyperbolic random graphs.
However, when omitting the probability distribution and looking at hyperbolic unit disk graphs as a graph class, little is known so far.

The class of hyperbolic unit disk graphs (\graphclass{HUDG}) has only been introduced recently~\cite{Blasius2022_StronglyHyperbolic}\footnote{
    We note that there are earlier results on a related family of graph classes parameterized by the disk size by Kisfaludi-Bak~\cite{KisfaludiBak2020_HypIntersec}.
    In a sense, the class \graphclass{HUDG} is the union of all these classes.
    This subtle difference is important when considering asymptotic behavior as it can be desirable to grow the disk size with the graph size; see \cite{Blasius2022_StronglyHyperbolic} for a detailed discussion.
}.
When choosing disks of small radius, the difference between Euclidean and hyperbolic geometry becomes negligible; also see our interactive visualization\footnote{
    \url{https://thobl.github.io/hyperbolic-unit-disk-graph}
} and Figure~\ref{fig:unit_disk_graphs}.

Arguably the most fundamental algorithmic question when it comes to studying graph classes is the computational complexity of the \emph{recognition problem}, i.e., \problemname{Recog(\graphclass{HUDG})} is the problem of testing whether a given graph is part of \graphclass{HUDG}.
In this paper we prove that \problemname{Recog(\graphclass{HUDG})} is \ER-complete.
Containment in~\ER is less obvious than in the Euclidean plane as distances are not (square roots of) a polynomial in hyperbolic geometry.
Nonetheless, containment is easy to show when using the hyperboloid model of the hyperbolic plane.
For \ER-hardness, our proof consists of five steps switching back and forth between Euclidean and hyperbolic variants of problems in a particular way.
Our proof has framework-character in the sense that the first three steps are independent of the specific problem and the remaining steps can probably be translated to other problems.
Thus we believe that this can be a template for proving \ER-hardness for other hyperbolic problems that have an \ER-hard Euclidean counterpart.
For our framework, we in particular use the Beltrami-Klein model of the hyperbolic plane to observe that \SimpleStretchability is equivalent in Euclidean and hyperbolic geometry in the sense that a pseudoline arrangement is stretchable in the Euclidean plane if and only if it is stretchable in the hyperbolic plane.

\begin{figure}[tb]
    \centering
    \includegraphics[width=0.45\linewidth]{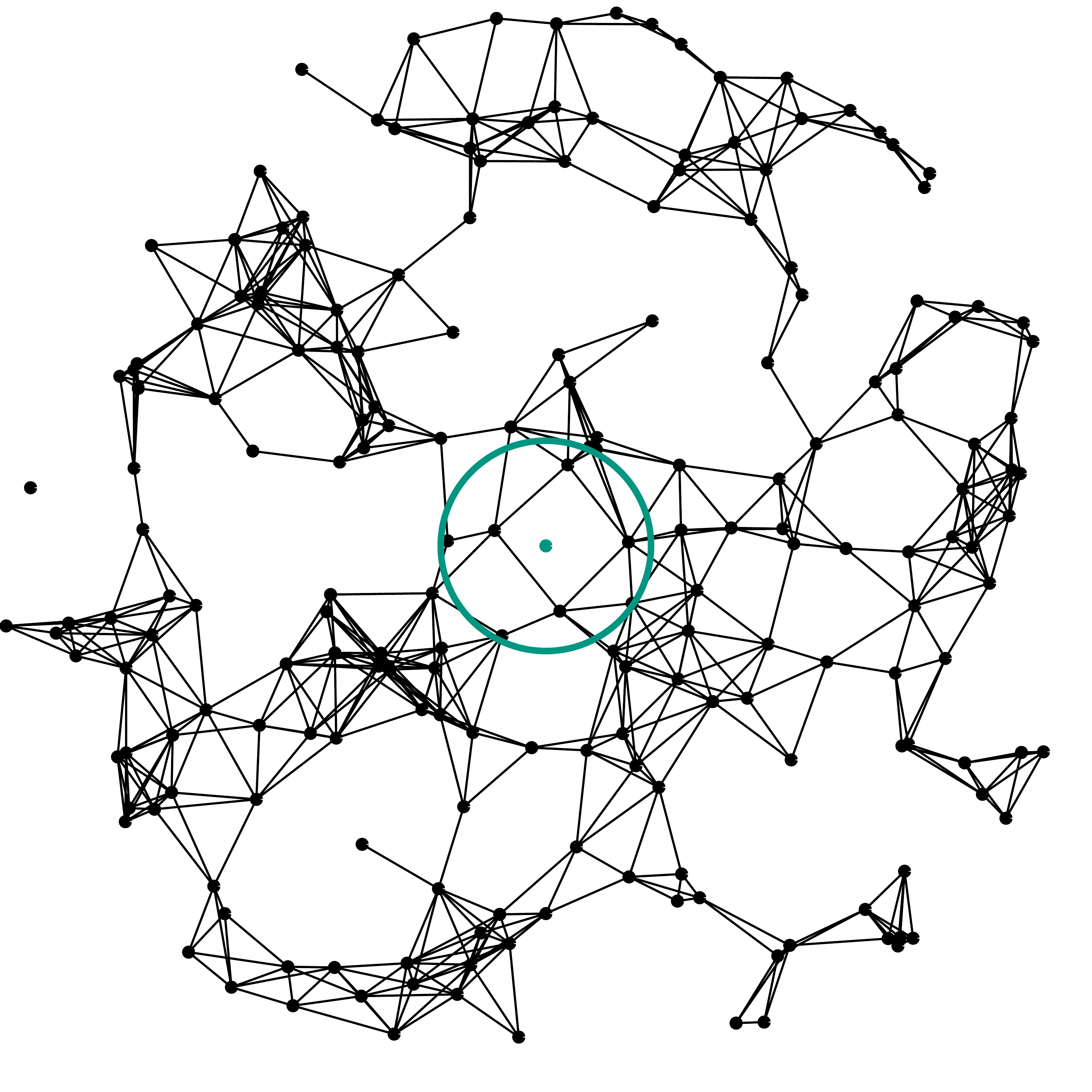}
    \hfill
    \includegraphics[width=0.45\linewidth]{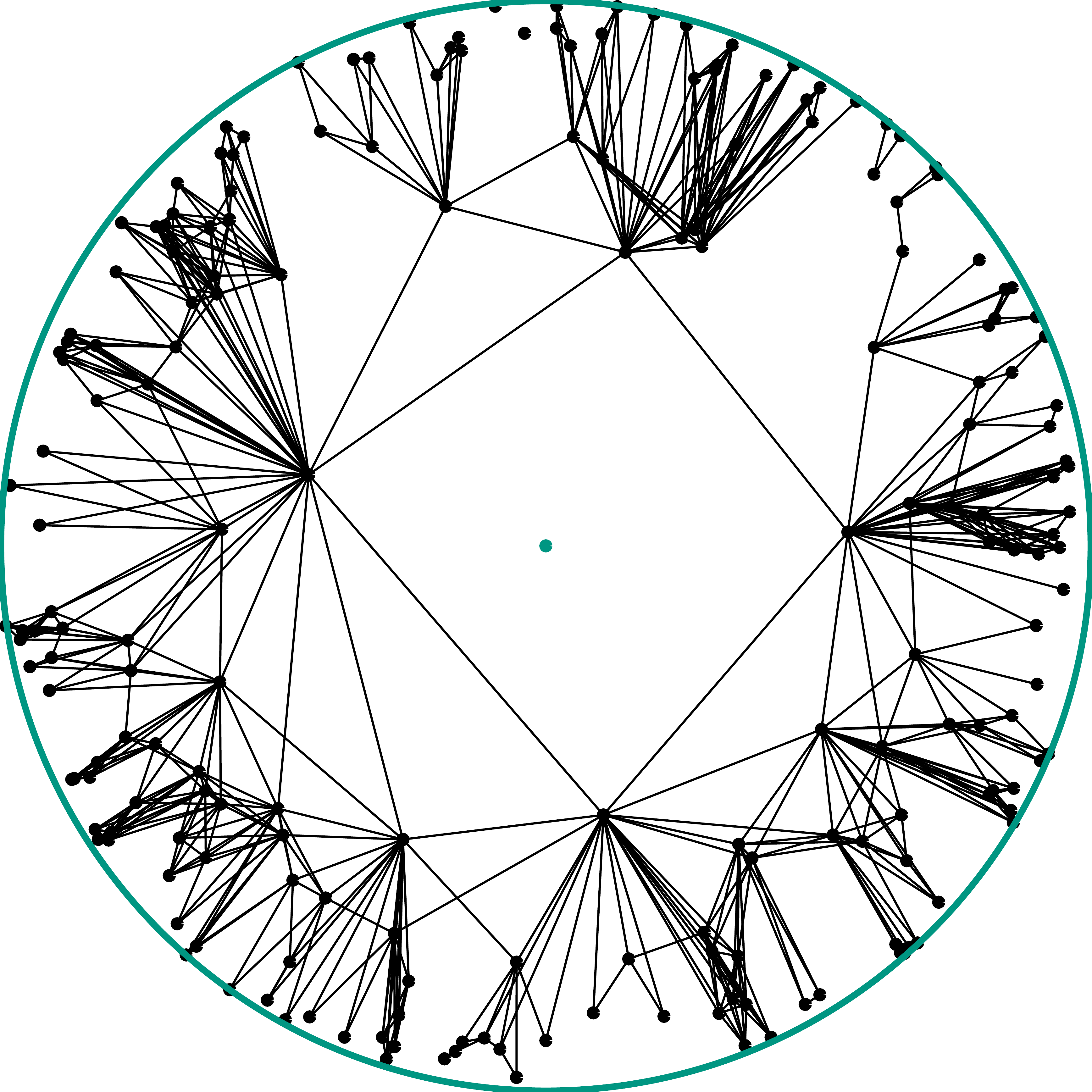}
    \caption{
        Two hyperbolic unit disk graphs.
        Thee green circle indicates the threshold distance below which vertices are connected.
        A small threshold (left) yields structures similar to Euclidean unit disk graphs.
        A large threshold (right) facilitates heterogeneous vertex degrees.
        }
    \label{fig:unit_disk_graphs}
\end{figure}

\subsection{Existential Theory of the Reals}

The \emph{existential theory of the reals} is the set of all true sentences of the form $\exists X \in \R^n : \varphi(X)$, where~$\varphi(X)$ is a quantifier-free formula consisting of polynomial equations and inequalities, e.g.\ $\exists X, Y \in \R : X \cdot Y = 6 \land X+Y = 5$.
We denote the decision problem whether such a sentence is true by~\ETR (which also stands for \enquote{existential theory of the reals}) and define the complexity class~\ER to contain all decision problems that polynomial-time reduce to~\ETR.
It holds $\NP \subseteq \ER \subseteq \PSPACE$~\cite{Canny1988_PSPACE}.
The class~\ER has gained increasing attention in the computational geometry community over the last years as it exactly captures the complexity of many geometry problems like the art gallery problem~\cite{Abrahamsen2022_ArtGallery_JACM}, geometric packing~\cite{Abrahamsen2020_Packing} or the recognition of many classes of geometric intersection graphs~\cite{Kratochvil1994_SegmentIntersectionGraphs,Matousek2014_IntersectionGraphsER,Schaefer2010_GeometryTopology}.

\subsection{Hyperbolic Geometry}

The are several ways to embed the hyperbolic plane into Euclidean space.
In this paper we use the \emph{Beltrami-Klein model} and the \emph{hyperboloid model}.

\subparagraph{Beltrami-Klein Model}
In the Beltrami-Klein model the hyperbolic plane~$\H^2$ is represented by the interior of a unit disk~$D$ in~$\R^2$ (the boundary of~$D$ is not part of the model).
The set of hyperbolic lines is exactly the set of chords of~$D$.
See Figure~\ref{fig:hyperbolic_models} (left).
\begin{figure}[tb]
    \centering
    \includegraphics{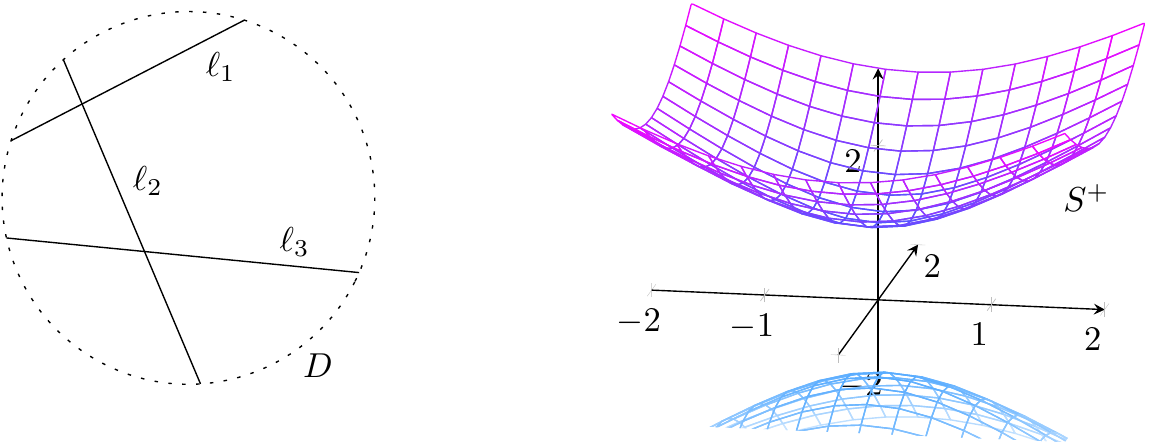}
    \caption{
        Left:~The Beltrami-Klein disk with three hyperbolic lines.
        Right:~The upper sheet~$S^+$ used for the hyperboloid model.
    }
    \label{fig:hyperbolic_models}
\end{figure}

\subparagraph{Hyperboloid Model}
Here the hyperbolic plane gets embedded into~$\R^3$.
The \emph{Minkowski quadratic form} $Q(x,y,z) := z^2 - x^2 - y^2$ defines a two-sheeted hyperboloid $S := \{(x,y,z) \in \R^3 \mid Q(x,y,z) = 1\}$, see Figure~\ref{fig:hyperbolic_models} (right).
The hyperbolic plane is represented by all points on the forward sheet~$S^+$ of~$S$, obtained by additionally requiring that~$z > 0$.
The hyperbolic distance between two points $u,v \in S^+$ is
\[
    d_h(u,v) = \arccosh(B(u,v))
    \qquad
    \text{with}
    \quad
    B(u, v) := u_zv_z - u_xv_x - u_yv_y
\]
where $B(u, v)$ is known as the \emph{Minkowski bilinear form} and $\arccosh(x) := \ln\bigl(x + \sqrt{x^2 - 1}\bigr)$ is the inverse hyperbolic cosine.
Note that the term inside the~$\arccosh(\cdot)$ is a polynomial.

\section{Simple Stretchability in the Euclidean and the Hyperbolic Plane}

An \emph{pseudoline arrangement}~$\A$ is a collection of \emph{pseudolines} ($x$-monotone curves in~$\R^2$) such that each pair of curves intersects at most once.
We assume that each pseudoline~$\ell \in \A$ is oriented and thus divides the plane~$\R^2$ into two open half-planes~$\ell^-$ and~$\ell^+$.
Further, $\A$ partitions the plane into \emph{cells}, i.e., maximal connected components of~$\R^2 \setminus \mathcal{A}$ not on any pseudoline.
We say that~$\A$ is \emph{simple} if any two lines intersect exactly once and no three lines intersect in the same point.
Given a pseudoline arrangement~$\A = \{\ell_1, \ldots, \ell_n\}$ we assign to each~$p \in \R^2$ a \emph{sign vector}~$\sigma(p) = (\sigma_i(p))_{i = 1}^n \in \{-,0,+\}^n$, where
\[
    \sigma_i(p) :=
    \begin{cases}
        - & \text{if } p \in \ell_i^- \\
        0 & \text{if } p \in \ell_i \\
        + & \text{if } p \in \ell_i^+
    \end{cases}
    \text{.}
\]
The \emph{combinatorial description}~$\D$ of~$\A$ is then given by~$\{\sigma(p) \mid p \in \R^2\}$.
We say that~$\A$ realizes~$\D$.
A pseudoline arrangement is \emph{stretchable} if there is a line arrangement with the same combinatorial description.
Not every pseudoline arrangement is stretchable and, given a combinatorial description~$\D$, deciding whether~$\D$ is stretchable is known as the \Stretchability problem (or \SimpleStretchability if~$\D$ is simple).
\Stretchability and \SimpleStretchability are famously known to be \ER-complete~\cite{Mnev1988_UniversalityTheorem,RichterGebert2002_UniversalityTheorem,Shor1991_Stretchability}.
\SimpleStretchability is the starting problem for many \ER-hardness reductions, e.g.~\cite{Bienstock1991_CrossingNumber,Hoffmann2017_PlanarSlopeNumber,Kratochvil1994_SegmentIntersectionGraphs,Schaefer2010_GeometryTopology,Schaefer2013_GraphRealizability}.

Apart from  line arrangements in the Euclidean plane~$\R^2$ one might also consider line arrangements in the hyperbolic plane~$\H^2$.
The main result of this section is that \SimpleStretchability is equivalent in Euclidean and hyperbolic geometry.

\begin{proposition}
    \label{prop:simple_stretchability_equivalence}
    Let~$\D$ be a combinatorial description of a simple pseudoline arrangement.
    Then there is a line arrangement realizing~$\D$ in~$\R^2$ if and only if there is one in~$\H^2$.
\end{proposition}

\begin{proof}
    The proof is an easy application of the Beltrami-Klein model of the hyperbolic plane.
    Given a Euclidean line arrangement, we can obtain a hyperbolic line arrangement with the same combinatorial description and vice versa, see Figure~\ref{fig:stretchability_equivalence}.
    \begin{figure}[tb]
        \centering
        \includegraphics{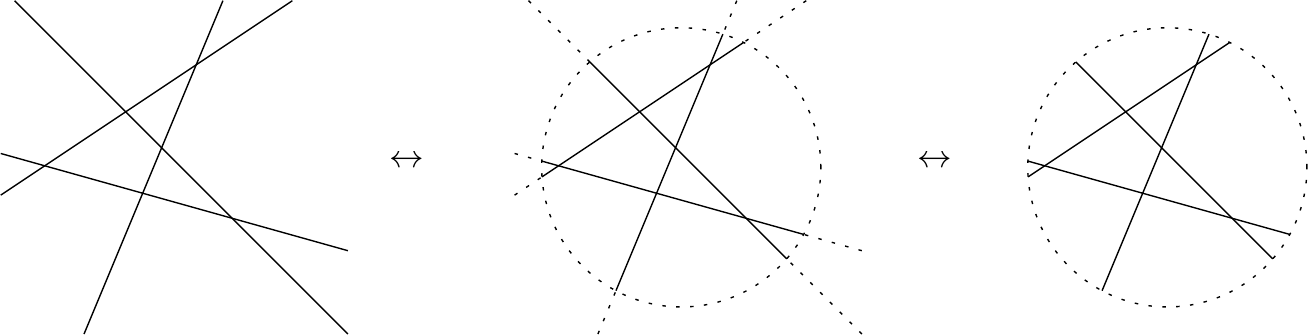}
        \caption{Transforming line arrangements between Euclidean and hyperbolic geometry.}
        \label{fig:stretchability_equivalence}
    \end{figure}

    Let~$\A_\R$ be a simple line arrangement in~$\R^2$ and~$D$ be a disk strictly enclosing all intersections of~$\A_\R$.
    For each line in~$\A_\R$, keep only its part inside~$D$.
    We think of~$D$ as a unit disk and obtain a representation of a hyperbolic line arrangement in the Beltrami-Klein model.
    
    For the other direction let~$\A_\H$ be a simple hyperbolic line arrangement and take a representation inside the Beltrami-Klein disk~$D$, so all hyperbolic lines are chords~$D$.
    Remove~$D$ and extend all chords to lines.
    The resulting Euclidean line arrangement has the same combinatorial description~$\D$ because~$\A_\H$ was simple:
    All possible intersections between two lines were already inside the Beltrami-Klein disk~$D$.
\end{proof}

\begin{remark}
    Proposition~\ref{prop:simple_stretchability_equivalence} is only about simple (pseudo)line arrangements.
    There is no corresponding result for the general (non-simple) \Stretchability problem:
    For example, given three lines~$\ell_1, \ell_2, \ell_3 \subseteq \H^2$, lines~$\ell_2$ and~$\ell_3$ may cross each other while both being parallel to~$\ell_1$.
    However, Proposition~\ref{prop:simple_stretchability_equivalence} may be extended to line arrangements where each pair of lines is still required to cross but multiple lines are allowed to cross at the same point.
\end{remark}

\section{The Framework}
Let~$\Pi_\R$ be a geometric decision problem for which \ER-hardness is shown in Euclidean geometry by a polynomial-time reduction~$f$ from (Euclidean) \SimpleStretchability.
We denote by~$\Pi_\H$ the corresponding decision problem obtained by considering the hyperbolic plane~$\H^2$ instead of the Euclidean plane~$\R^2$.
Our framework below consists of several (hopefully) simple steps that allow us to prove \ER-hardness of~$\Pi_\H$ by using the reduction for~$\Pi_\R$:

\begin{enumerate}
    \item\label{itm:hyperbolic_instance}
    Let~$\D$ be an instance of \SimpleStretchability in~$\H^2$, i.e., a combinatorial description of a simple pseudoline arrangement.
    \item\label{itm:euclidean_instance}
    Use Proposition~\ref{prop:simple_stretchability_equivalence} to consider~$\D$ to be an instance of \SimpleStretchability in~$\R^2$.
    \item\label{itm:euclidean_reduction}
    Use the reduction~$f$ to obtain an instance~$I = f(\D)$ of~$\Pi_R$ equivalent to~$\D$.
    \item\label{itm:yes_instances}
    Prove that every yes-instance of~$\Pi_\R$ is also a yes-instance of~$\Pi_\H$.
    \item\label{itm:no_instances}
    Prove that a line arrangement realizing~$\D$ can be extracted from a realization of~$I$ in~$\H^2$.
\end{enumerate}

Steps~\ref{itm:hyperbolic_instance}, \ref{itm:euclidean_instance} and~\ref{itm:euclidean_reduction} require no work when applying the framework.

Step~\ref{itm:yes_instances} ensures that a stretchable instance~$\D$ yields a yes-instance of~$\Pi_\H$.
This step requires to come up with a new argument but we expect it to be relatively simple because locally~$\R^2$ and~$\H^2$ are very similar.
A promising approach is to scale a Euclidean realization of~$I$ to a tiny area and then interpret the Euclidean polar coordinates as hyperbolic ones.

Step~\ref{itm:no_instances} ensures correctness.
By showing that a line arrangement realizing~$\D$ can be extracted from a realization of~$I$ in~$\H^2$ we show that a no-instance~$\D$ maps to a no-instance of~$\Pi_\H$.
Reduction~$f$ might help us again here (though not as a black box as in Step~\ref{itm:euclidean_reduction}):
If we are lucky, the argument why a realization of~$I$ in~$\R^2$ induces a Euclidean line arrangement realizing~$\D$ only uses the axioms of \emph{absolute geometry} (the common \enquote{subset} of Euclidean and hyperbolic geometry) and works without any adaptations for realizations in~$\H^2$, too.

\section{Recognition of Hyperbolic Unit Disk Graphs}
\label{sec:recog_hyperbolic_unit_disk_graphs}

We apply our framework to prove that~$\problemname{Recog(\graphclass{HUDG})}$, the recognition problem of hyperbolic unit disk graphs, is \ER-hard.
For Euclidean geometry this is shown in~\cite{Kang2012_BallGraphs,McDiarmid2010_UnitDiskGraphs,McDiarmid2013_UnitDiskGraphs}.
Let us note that \graphclass{UDG} and \graphclass{HUDG} are not the same:
For example, a star graph with six leaves is a hyperbolic unit disk graph but not a Euclidean one.

For Step~\ref{itm:hyperbolic_instance} of our framework let~$\D$ be an instance of \SimpleStretchability in~$\H^2$.
We consider it to be an equivalent instance in~$\R^2$ for Step~\ref{itm:euclidean_instance}.
In Step~\ref{itm:euclidean_reduction} we use the reduction~$f$ from the literature proving that \problemname{Recog(\graphclass{UDG})} in~$\R^2$ is \ER-hard~\cite{Kang2012_BallGraphs,McDiarmid2010_UnitDiskGraphs,McDiarmid2013_UnitDiskGraphs}.
We obtain a graph~$G_\D$ that is a Euclidean unit disk graph if and only if~$\D$ is stretchable.

\medskip
Though not required for the framework, let us shortly summarize the reduction~$f$ to construct~$G_\D$ from~$\D$ as given in~\cite{McDiarmid2010_UnitDiskGraphs}.
Let~$n$ be the number of pseudolines~$\ell_1, \ldots, \ell_n$ and $m = 1 + \binom{n+1}{2}$ be the number of cells~$C_1, \ldots, C_m$.
The arrangement described by~$\D$ has exactly this number of cells, because it is simple.
We define~$G_\D$ to be the graph with vertex set $V = A \cup B \cup C$ for $A = \{a_1, \ldots, a_n\}$, $B = \{b_1, \ldots, b_n\}$ and $C = \{c_1, \ldots, c_m\}$.
Here we assume that vertex~$c_i$ corresponds to cell~$C_i$.
For the edges, each of the sets~$A,B,C$ forms a clique.
Further, each~$a_i \in A$ (for~$i \in \{1, \ldots, n\}$) is connected to~$c_j$ (for~$j \in \{1, \ldots, m\}$) if and only if~$C_j \in \ell_i^-$.
Similarly, each~$b_i \in B$ is connected to~$c_j$ if and only if~$C_j \in \ell_i^+$.

\medskip
For Step~\ref{itm:yes_instances} we have to show that every Euclidean unit disk graph is also a hyperbolic unit disk graph.
This has recently been proven by Bl{\"{a}}sius, Friedrich, Katzmann and Stephan:

\begin{lemma}[\cite{Blasius2022_StronglyHyperbolic}]
    \label{lem:udg_subset_hudg}
    Every Euclidean unit disk graph is also a hyperbolic one, so $\graphclass{UDG} \subseteq \graphclass{HUDG}$.
\end{lemma}

As foreshadowed above, the proof scales a Euclidean unit disk intersection representation to a tiny area until the Euclidean and hyperbolic plane are \enquote{similar enough}.
Then the polar coordinates in~$\R^2$ can be used as polar coordinates in~$\H^2$ without changing any adjacencies.

For Step~\ref{itm:no_instances} it remains to prove how a line arrangement realizing~$\D$ in~$\H^2$ can be extracted from a realization of~$G_\D$ in~$\H^2$.

\begin{lemma}[{adapted from~\cite[Lemma~$1$]{McDiarmid2010_UnitDiskGraphs}}]
    \label{lem:extract_line_arrangement}
    Given a realization of~$G_\D$ as the intersection graph of equally sized disks in~$\H^2$.
    Then the line arrangement~$L = \{\ell_1, \ldots, \ell_n\}$ defined by
    \[
        \ell_i := \{p \in \H^2 \mid \mathrm{d}(p, a_i) = \mathrm{d}(p, b_i)\}
    \]
    has combinatorial description~$\D$.
    Here~$\mathrm{d}(\cdot, \cdot)$ denotes the hyperbolic distance.
\end{lemma}

\begin{proof}
    The proof is exactly the same as the proof of Lemma~$1$ in~\cite{McDiarmid2010_UnitDiskGraphs} where McDiarmid and M{\"{u}}ller prove that taking the perpendicular bisectors of the segments between any pair of points~$a_i$ and~$b_i$ yields a Euclidean line arrangement realizing~$\D$.
    Their argument works in~$\H^2$ by just replacing Euclidean distances with hyperbolic distances.
\end{proof}

At this point we proved \ER-hardness of \problemname{Recog(\graphclass{HUDG})}.
To get \ER-completeness we prove \ER-membership next.

\begin{lemma}
    \label{lem:recog_hyperbolic_unit_disk_membership}
    Recognizing hyperbolic unit disk graphs is in \ER.
\end{lemma}

\begin{proof}
    By a result from Erickson, van der Hoog and Miltzow we can prove \ER-membership by describing a polynomial-time verification algorithm for a real RAM machine\footnote{
        The real RAM extends the classical word RAM by additional registers that contain real numbers (with arbitrary precision).
        The basic arithmetic operations $+$, $-$, $\cdot$ and $/$ are supported in constant time.
        However, arbitrary analytic functions (like $\arccosh$) are not supported.
        See~\cite{Erickson2022_SmoothingTheGap} for a formal definition.
    }~\cite{Erickson2022_SmoothingTheGap}.
    Given a graph~$G = (V,E)$ and for each vertex~$v \in V$ a point $(v_x, v_y, v_z)$ in the hyperboloid model of the hyperbolic plane representing the center of an equal-radius disk.
    Compute $d_{\text{adj}} := \max_{uv \in E} B(u,v)$ and $d_{\text{non-adj}} := \min_{uv \not\in E} B(u,v)$ where~$B(\cdot,\cdot)$ is the Minkowski bilinear form.
    $B(\cdot,\cdot)$ is a polynomial, so it is computable on a real RAM.
    We can think of $d_{\text{adj}}$ and $d_{\text{non-adj}}$ as distances in hyperbolic space (actually they are the hyperbolic cosine of a distance), but since~$\arccosh$ is a monotone function, this view is justified.
    Now if and only if $d_{\text{adj}} < d_{\text{non-adj}}$, then there is a radius~$r$ such that~$G$ is a hyperbolic unit disk graph with radius~$r$ (choose~$r$ such that $d_{\text{adj}} \leq \frac{\cosh(r)}{2} < d_{\text{non-adj}}$).
    The algorithm takes~$O(\abs{V}^2)$ time. This is polynomial in the input size, proving \ER-membership.
\end{proof}

We conclude with the following theorem:

\begin{theorem}
    \label{thm:recog_hyperbolic_unit_disk_er_complete}
    Recognizing hyperbolic unit disk graphs is \ER-complete.
\end{theorem}

\section{Conclusion and Outlook}

We presented a simple framework that allows us to translate \ER-hardness reductions for geometric decision problems in~$\R^2$ into reductions for their counterparts~$\H^2$.
As an application we proved that \problemname{Recog(\graphclass{HUDG})} is \ER-complete.
Promising candidates for further applications of our framework are the recognition of unit ball graphs (i.e., a generalization of our result to higher dimensions) as already done in~$\R^d$ in~\cite{Kang2012_BallGraphs} or \problemname{Recog(\graphclass{CONV})}, the recognition problem for intersection graphs of convex sets (Euclidean reduction is in~\cite{Schaefer2010_GeometryTopology}).

Technically, the framework also works for the recognition problems \problemname{Recog(\graphclass{HSEG})} and \problemname{Recog(\graphclass{HDISK})}, where \graphclass{(H)SEG} and \graphclass{(H)DISK} denote the classes of intersection graphs of (hyperbolic) segments and disks, respectively (Euclidean reductions are in~\cite{Kang2012_BallGraphs,Kratochvil1994_SegmentIntersectionGraphs,Matousek2014_IntersectionGraphsER,McDiarmid2013_UnitDiskGraphs,Schaefer2010_GeometryTopology}).
However, these are not really interesting as $\graphclass{SEG} = \graphclass{HSEG}$ (easy to see in the Beltrami-Klein model) and $\graphclass{DISK} = \graphclass{HDISK}$ (easy to see in the \Poincare model, not considered here).
Therefore \ER-completeness for \problemname{Recog(\graphclass{HSEG})} and \problemname{Recog(\graphclass{HDISK})} follows directly from the Euclidean cases.
Other interesting problems to consider in~$\H^2$ are linkage realizability~\cite{Abel2016_PlaneRigidity,Schaefer2013_GraphRealizability}, simultaneous graph embeddings~\cite{Cardinal2015_SimultaneousEmbedding,Kyncl2011_AbstractTopologicalGraphs} or RAC-drawings~\cite{Schaefer2021_RAC}.

\subparagraph{Acknowledgements}
We thank Torsten Ueckerdt for discussion on proving the membership of \problemname{Recog(\graphclass{HUDG})} in \ER.

\bibliographystyle{plainurl}
\bibliography{references}

\end{document}